\documentclass[12pt]{article}

\usepackage{xspace}
\usepackage{amsmath}
\usepackage{amsthm}
\usepackage{color}
         \usepackage{fullpage}
\usepackage{amsfonts}

\newcommand{\DRCC}{\textbf{DCBC}\xspace}

\newcommand{\CBMC}{\textbf{CBC}\xspace}
\newcommand{\DBOM}{\textbf{DBNS}\xspace}
\newcommand{\BOM}{\textbf{BNS}\xspace}
\newcommand{\DSteinerT}{\textbf{DST}\xspace}
\newcommand{\FDSteinerT}{\textbf{FDST}\xspace}
\newcommand{\CSC}{\textbf{CSC}\xspace}
\newcommand{\GST}{\textbf{GST}\xspace}
\newcommand{\DCSC}{\textbf{DCSC}\xspace}
\newcommand{\DGST}{\textbf{DGST}\xspace}

\newtheorem{theorem}{Theorem}[section]
\newtheorem{corollary}[theorem]{Corollary}
\newtheorem{lemma}[theorem]{Lemma}
\newtheorem{claim}{Claim}[section]

\title{An Approximation Algorithm for the Connected Maximum Coverage Problem in Directed Graphs}

\author{Gianlorenzo D'Angelo\\
Gran Sasso Science Institute (GSSI)\\
Viale F. Crispi, 7, 67100 - L'Aquila, Italy\\
\texttt{gianlorenzo.dangelo@gssi.it}
\and
Esmaeil Delfaraz\\
University of L'Aquila\\
Via Vetoio snc, 67100 - L'Aquila, Italy\\
\texttt{esmaeil.delfarazpahlevanloo@univaq.it}
}

\begin{document}

\maketitle

\begin{abstract}
In the Directed rooted Connected Budgeted maximum Coverage problem (\DRCC), we are given a collection of subsets $\mathcal{S}$, defined over a ground set $X$, and a directed graph $G=(V,E)$, where each node is associated with a set of $\mathcal{S}$. Each set in $\mathcal{S}$ has a different cost and each element of $X$ gives a different prize. The goal is to find a subcollection $\mathcal{S}'\subseteq \mathcal{S}$ such that $\mathcal{S}'$ induces an out-tree rooted at a given node, the total cost of the sets in $\mathcal{S}'$ does not exceed a budget $B$, and the total prize of the elements covered by $\mathcal{S}'$ is maximized.

In this paper, we provide an algorithm for \DRCC that guarantees an approximation ratio of $O\left(\frac{\sqrt{|V|}\log^2|X|}{\epsilon^2}\right)$, with a budget violation of a factor $1+\epsilon$, where $\epsilon\in (0,1]$. 
Our algorithm also implies an improved approximation factor for the budgeted node-weighted Steiner problem in directed graphs, a particular case of \DRCC where the prize function is additive, for which we improve from $O\left(\frac{1}{\epsilon^2}|V|^{2/3}\log|V|\right)$ to $O\left(\frac{1}{\epsilon^2}|V|^{1/2}\log^2|V|\right)$. 
\end{abstract}

\section{Introduction}
In the \emph{budgeted maximum coverage} problem, we are given a ground set $X$ of elements with associated prizes, a collection $\mathcal S$ of subsets of $X$ with associated costs, and a budget $B$. The aim is to find a subcollection $\mathcal{S}'\subseteq \mathcal{S}$ such that the total cost of the sets in $\mathcal{S}'$ does not exceed $B$ and the total prize of the elements covered by $\mathcal{S}'$ (i.e., $\bigcup_{S\in \mathcal{S}'} S$) is maximized~\cite{khuller1999budgeted}.
The \emph{Connected Budgeted maximum Coverage} problem (\CBMC) is a generalization of the budgeted maximum coverage problem in which the sets in $\mathcal{S}$ are associated with the nodes of a graph $G=(V,A)$ and the subcollection $\mathcal{S}'$ must induce a connected subgraph $T$ in $G$.

The \CBMC problem is motivated by several applications in multi-agent path planning, wireless sensor networks, and bioinformatics. Consider, for example, the exploration of a region through an Unmanned Aerial Vehicle that can take pictures of an area close to its current location. In this scenario, the areas to be explored correspond to the elements in $X$, and the nodes of the graph correspond to the locations that the vehicle can reach. Computing a connected maximum coverage corresponds to finding a small set of connected locations that allow the exploration of the largest possible area~\cite{chen2020optimal,xu2021throughput}. 
Another application is the deployment of wireless sensor networks in a scenario where each sensor is able to detect a set of target points in its sensing range, and one wants to deploy a bounded set of connected sensors that detects the largest number of target points~\cite{yu2019connectivity}. Here, the sensor network corresponds to the graph $G$, and the target points to the elements $X$.
Vandin et al.~\cite{vandin2011algorithms} studied \CBMC motivated by the detection of driver mutations in protein-to-protein interaction networks. In these networks, a node represents a protein, and an edge represents an interaction between two proteins. Each protein is associated with a gene mutation and a set of cancer patients who are affected by such mutation.  It is widely believed that cancer is associated with a connected series of mutations in these networks, called pathways~\cite{Hahn2002Modelling}. Therefore, finding a connected set of $B$ nodes with maximum coverage corresponds to finding the $B$ connected mutations that affect the largest number of cancer patients. Variants of \CBMC find application in network surveillance~\cite{livne2023optimally}, recovery of power networks~\cite{guha1999efficient}, and data acquisition~\cite{yank2024budgeted}.

The \emph{Directed rooted Connected Budgeted maximum Coverage} problem (\DRCC) is a generalization of \CBMC to directed graphs, where the aim is to find a rooted out-tree maximizing the prizes of covered elements and respecting a budget constraint. Besides the same applications as \CBMC in the cases where the underlying network is directed, the \DRCC problem has specific motivating application scenarios in facility location, epidemiology, and computational social choice.
Consider a large warehouse from which goods are delivered through a road network to smaller warehouses or retail shops that serve a set of customers. Here, the graph models the (directed) road network; the ground set represents the set of customers; and shops, represented as nodes in the graph, are associated with the set of customers that they may serve. In order to maximize the number of customers that can be served at a given delivery cost, one needs to compute an out-tree, rooted at the node representing the warehouse, that maximizes the overall number of covered customers and satisfies the budget constraint.
Connectivity is required because goods are distributed from the main warehouse to the open shops via a directed road network. 
Vehicles departing from the warehouse (root) reach the retail shops (selected nodes) via directed paths.
Other applications require computing rooted out-trees in directed graphs in order to reconstruct epidemic outbreaks~\cite{Mishra2023Reconstructing,Rozenshtein2016Reconstructing} or to maximize the voting power of a voter in liquid democracy~\cite{dangelo2022computation}.

\subsection*{Related work}
Problems \CBMC and \DRCC have already been studied under the lens of approximation algorithms.
However, the state-of-the-art algorithms achieve approximation ratios that are, in the worst case, linear in $|V|$~\cite{hochbaum2020approximation,ran2016approximation,vandin2011algorithms} or depend on the budget $B$~\cite{hochbaum2020approximation,d2022budgeted}, and, in some cases, only work under specific assumptions~\cite{vandin2011algorithms,hochbaum2020approximation,ran2016approximation}.

Vandin et al.~\cite{vandin2011algorithms} considered the special case of \CBMC where the cost is equal for all nodes and provided a $c\rho$-approximation, where $c=(2e-1)/(e-1)$ and $\rho$ is the radius of the connected subgraph induced by an optimal solution. Hochbaum and Rao~\cite{hochbaum2020approximation} improved this bound to $\min\{((1-1/e)(1/\rho-1/B))^{-1}, B\}$. This latter result also holds true in the more general case in which the prize function is a monotone submodular function of the set of nodes in the graph. Still, the cost function is assumed to be constant for all nodes.
Ran et al.~\cite{ran2016approximation} considered \CBMC without restrictions on the cost function but under the assumption that if two sets in $\mathcal{S}$ have a non-empty intersection, then the corresponding nodes in $G$ must be adjacent. In this setting, they provided an $O(\Delta\log{|V|})$-approximation algorithm, where $\Delta$ is the maximum degree of $G$. D'Angelo et al.~\cite{d2022budgeted} improved this result to $O(\log{|V|})$ under the same assumption. Moreover, they gave an $O(\sqrt{B})$-approximation algorithm for \DRCC.
In the worst case, $\rho=\Omega(|V|)$ and $\Delta = \Omega(|V|)$, and thus the approximation ratio of algorithms in~\cite{hochbaum2020approximation,ran2016approximation,vandin2011algorithms} are linear in $|V|$. Moreover, $B$ can be exponential in $|V|$ for general cost functions.

The generalization of \CBMC in which the prize function is a monotone submodular function on the set of nodes in the graph and the cost function on nodes is defined over positive integers has been studied by Kuo et al.~\cite{kuo2014maximizing}, who gave an $O(\Delta \sqrt{B})$-approximation algorithm. For the same problem, D'Angelo et al.~\cite{d2022budgeted} gave an $O(\sqrt{B})$-approximation algorithm, which also applies to the directed case with the same bound. 
They also considered the rooted variant of the same problem in which a specific root node is required to belong to the solution and provided an $O(\frac{1}{\epsilon^{3}}\sqrt{B})$-approximation algorithm, if a budget violation of a factor $1+\epsilon$, for some $\epsilon\in (0,1]$, is allowed.
Ghuge and Nagarajan~\cite{ghuge2020quasi} provided a tight quasi-polynomial time $O(\frac{\log n'}{\log \log n'})$-approximation algorithm for the directed case, where $n'$ is the number of nodes in an optimal solution.

The \emph{Directed Budgeted Node-weighted Steiner} problem (\DBOM) is a particular case of \DRCC in which both costs and prizes are associated with the nodes of a directed graph, and the goal is to find an out-tree rooted at a specific node that satisfies a budget constraint on the sum of costs and maximizes the sum of prizes of the nodes. For \DBOM, D'Angelo and Delfaraz~\cite{d2022approximation} gave an $O\left(\frac{1}{\epsilon^2}|V|^{2/3}\log{|V|}\right)$-approximation algorithm which violates the budget constraint by a factor of at most $1+\epsilon$, for any $\epsilon \in (0, 1]$.
 The \emph{Budgeted Node-weighted Steiner} problem (\BOM) is the \DBOM problem restricted to undirected graphs and can be seen as a particular case of \CBMC in which both costs and prizes are associated with the nodes of an undirected graph, and the goal is to find a tree
that satisfies a budget constraint on the sum of costs and maximizes the sum of prizes of the nodes. For this problem, Guha et al.~\cite{guha1999efficient} gave a polynomial time $O(\log^2 |V|)$-approximation algorithm that violates the budget constraint by a factor of at most $2$.
Moss and Rabani~\cite{moss2007approximation} improved the approximation factor to $O(\log |V|)$, with the same budget violation. Later, Bateni et al.~\cite{bateni2018improved} proposed an $O(\frac{1}{\epsilon^{2}}\log |V|)$-approximation algorithm which requires a budget violation of only
 $1+\epsilon$, for any $\epsilon\in (0,1]$.
Kortsarz and Nutov~\cite{kortsarz2009approximating} showed that this problem admits no $o(\log \log |V|)$-approximation algorithm, unless $NP \subseteq DTIME(n ^{\text{polylog}(n)})$, even if the algorithm is allowed to violate the budget constraint by a factor equal to a universal constant. Bateni et al.~\cite{bateni2018improved} showed that the integrality gap of the standard flow-based LP relaxation for the budgeted node-weighted Steiner is unbounded if no budget violation is allowed. \CBMC also generalizes the \emph{budgeted connected dominating set} problem~\cite{khuller2020analyzing} for which there exists a constant factor approximation algorithm~\cite{lamprou2020improved}.

The \emph{minimum Connected Set Cover} (\CSC) and the  \emph{Directed minimum Connected Set Cover} (\DCSC) problems are two minimization versions of \CBMC and \DRCC, where the aim is to cover \emph{all} the elements of the ground set with a minimum-cost tree or out-tree, respectively. In the \emph{node-weighted Group Steiner Tree} (\GST) and \emph{Directed node-weighted Group Steiner Tree} (\DGST), we are given a graph (directed graph, respectively) with costs associated with the nodes and a family of $k$ subsets of the nodes called groups. The aim is to find a minimum-cost tree (out-tree, respectively) that contains at least one node in each group.
Problems \CSC (\DCSC, respectively) and \GST (\DGST, respectively) are strictly related from the approximation point of view in the sense that there exists an $\alpha(|V|,|X|)$-approximation algorithm for \CSC if and only if there exists an $\alpha(|V|,k)$-approximation algorithm for \GST~\cite{Elbassioni12relation}. Both problems can be approximated by an approximation algorithm for the \emph{node-weighted Steiner tree problem in directed graphs} (\DSteinerT), where each element $e$ (group $g$, resp.) is represented as a node terminal with incoming edges from the nodes corresponding to the sets containing $e$ (or from the nodes belonging to $g$, resp.). An $\alpha(|V|,|R|)$-approximation algorithm for \DSteinerT, where $R$ is the set of terminal, implies an $\alpha(|V|,|X|)$-approximation algorithm for \CSC and \DCSC and an $\alpha(|V|,k)$-approximation algorithm for \GST and \DGST.

Problems \CSC and \GST have been initially motivated by applications in biology~\cite{debinski2000survey} and VLSI design~\cite{reich89beyond}, respectively, but find applications in also in other fields. For example, in the Watchman Route Problem, a widely studied problem in path finding~\cite{livne2023optimally,skyler2022solving}, we are given an undirected graph $G=(V,A)$, a cost function on the edges, a starting point $r$, and a \emph{line-of-sight} function $LOS:V\rightarrow 2^V$ which determines which nodes any given node can see. The aim is to compute a minimum-cost path $P$ starting from $r$ such that all the nodes in $V$ are in the line-of-sight of at least one node in $P$. The Watchman Route Problem can be approximated by an approximation algorithm for \CSC, in the special case where the ground set is made of all the nodes of $G$ and the $LOS$ function induces the family of subsets ($S_v=LOS(v)\cup \{v\}$). By standard techniques, we can move the edge costs to the nodes and transform a tree into a route by losing a 2-approximation factor.

Problems \CSC and \GST have been studied separately until Elbassioni et al.~\cite{Elbassioni12relation} made a connection between these two problems. Zhang et al.~\cite{zhang2009algorithms} gave two algorithms for \CSC with approximation ratios of $O(D_c\log |V|)$ and $O(D_c\log N)$, respectively, where $D_c$ is the length of the longest path in $G$ between two nodes corresponding to non-disjoint sets and $N$ the maximum size of a set. As observed in~\cite{Elbassioni12relation}, both bounds are $\Omega(|V|)$. Khandekar et al.~\cite{Khandekar2012} gave a $O(\sqrt{|V|}\log|V|)$-approximation algorithm for \CSC and \GST.

The variant of \GST where the costs are associated with the edges of the graph instead of the nodes has been widely investigated. 
Garg et al.~\cite{garg2000polylogarithmic} gave a randomized $O(\log N \log |V| \log\log|V|\log k)$-approximation algorithm, where $N$ is the size of the largest group.
They first gave a randomized $O(\log N\log k)$-approximation for the case when the input graph is a tree and then extended this result to general graphs by using probabilistic tree embeddings~\cite{bartal98approximating,fakcharoenphol2004}.
When the cost is associated with the nodes of the graph, such tree embeddings cannot be used, and hence, the algorithm by Garg et al. cannot be extended to \GST.
Naor et al.~\cite{naor2011online} gave a quasi-polynomial-time randomized algorithm for the \emph{online} version of \GST problem with a $O(\log^3|V|\log^7 k)$ competitive ratio and a polynomial time online algorithm for the edge-weighted version of \GST with competitive ratio of $O(\log^5|V|\log k)$.

Halperin et al.~\cite{halperin2007integrality} showed that the integrality gap of the standard flow-based linear relaxation of the edge-weighted version of \GST is $\Omega(\log^2 k/(\log\log k)^2)$. Halperin and Krauthgamer~\cite{halperin2003polylogarithmic} showed that, unless $P=NP$, this problem cannot be approximated within a factor $\Omega(\log^{2-\epsilon} k)$. Under stronger complexity assumptions, Grandoni et al.~\cite{grandoni2019log2} improved this factor to $\Omega(\frac{\log^2 k}{\log\log k})$.

\subsection*{Our results}
For \DRCC, we provide a bicriteria approximation algorithm that, for any $\epsilon\in (0,1]$, guarantees an approximation ratio of $O\left(\frac{\sqrt{|V|}\log^2{|X|}}{\epsilon^2}\right)$ at the cost of a violation in the budget constraint of a factor at most $1+\epsilon$ (see Theorem~\ref{thMainDBudget} in Section~\ref{sec:directed}). As a consequence, we improve the approximation ratio for the \DBOM from $O\left(\frac{1}{\epsilon^2}|V|^{2/3}\log{|V|}\right)$ to $O\left(\frac{1}{\epsilon^2}|V|^{1/2}\log^2{|V|}\right)$, in both cases with a budget violation of a factor at most $1+\epsilon$, for any $\epsilon \in (0, 1]$ (See Corollary~\ref{cor:MainDBudget} in Section~\ref{sec:directed}). 

Our algorithm for \DRCC uses, as a subroutine, an algorithm for the \emph{node-weighted Steiner tree problem in directed graphs} (\DSteinerT). In particular, our algorithm requires that such a subroutine computes a tree whose cost is within a bounded factor from the optimum of the \emph{standard flow-based linear programming relaxation} of \DSteinerT. Previous algorithms for \DSteinerT only focus on the approximation factor with respect to the \emph{optimum} of \DSteinerT but do not ensure a bounded factor over the optimum of its relaxation~\cite{charikar1999approximation}. Therefore, we introduce a new algorithm for \DSteinerT that guarantees this factor to be $O(\sqrt{|V|}\log{|V|})$ (see Theorem~\ref{thDSteinerTree} in Section~\ref{sec:DSteinerT}).
The combinatorial algorithm by Charikar et al.~\cite{charikar1999approximation} achieves a better approximation factor of $O(|R|^\epsilon)$, where $R$ is the set of terminals in the Steiner tree. Our result complements the result by Li and Laekhanukit~\cite{li2025polynomial}, who recently showed that the integrality gap of the standard linear program is $\Omega(|V|^{0.0418})$.

In a previous version of this paper, we erroneously claimed an improved approximation factor for \CBMC. Our algorithm was based on a flawed generalized analysis of Klein and Ravi's algorithm for the node-weighted Steiner tree problem~\cite{klein1995nearly}.

\section{Notation and Problem Statement}

For two integers $i, j$, let $[i, j]:=\{i,\ldots,j\}$
and $[i]:=[1,i]$. 
Let $G=(V, A)$ be a directed graph and $c:V \rightarrow \mathbb{R}^{\ge 0}$ be a nonnegative cost function on nodes. 

A \emph{path} is a directed graph made of a sequence of distinct nodes $(v_1, \dots, v_s)$ and a sequence of directed arcs $(v_i, v_{i+1})$, $i\in[s-1]$. An \emph{out-tree} (a.k.a. out-arborescence) is a directed graph in which there is exactly one directed path from a specific node $r$, called \emph{root}, to each other node. If a subgraph $T$ of a directed graph $G$ is an out-tree, then we say that $T$ is an out-tree of $G$.  For simplicity of reading, we will refer to out-trees simply as trees when it is clear that we are in the context of directed graphs. 
Given two nodes $u, v \in V$, the cost of a path from $u$ to $v$ in  $G$ is the sum of the cost of its nodes. 
A path from $u$ to $v$ with the minimum cost is called a \textit{shortest path} and its cost, denoted by $dist(u,v)$, is called the \textit{distance} from $u$ to $v$ in $G$.
A graph $G$ is called \emph{$B$-proper} for the node $r$ if $dist(r,v) \le B$ for any $v$ in  $V$.
For any subgraph $G'$ of $G$, we denote by $V(G')$ and $A(G')$ the set of nodes and arcs in $G'$, respectively. Given a subset of nodes $V' \subseteq V$, $G[V']$ denotes the subgraph of $G$ induced by nodes $V'$, i.e., $V(G[V']) = V'$ and $A(G[V'])=\{(u, v) \in A : u,v \in V'\}$).

Let $X$ be a ground set of elements, $\mathcal{S} \subseteq 2^X$ be a collection of subsets of $X$, and $p: X \rightarrow \mathbb{R}^{\ge 0}$ be a prize function over the elements of $X$. 
In the \emph{Directed rooted Connected Budgeted maximum Coverage} (\DRCC), each node $v$ of a directed graph $G$ is associated with a set $S_v$ of $\mathcal{S}$ and the goal is to find a rooted out-tree $T$ of $G$ with bounded cost that maximizes the overall prize of the union of the sets associated with the nodes in $T$.
Formally, in \DRCC we are given as input a ground set $X$, a collection  $\mathcal{S} \subseteq 2^X$ of subsets of $X$, a directed graph $G=(V,A)$, where each node $v\in V$ is associated with a set $S_v$ of $\mathcal{S}$, a root node $r\in V$, a cost function $c: V \rightarrow \mathbb{R}^{\ge 0}$ on the nodes of $G$, a prize function $p: X \rightarrow \mathbb{R}^{\ge 0}$ on the ground set $X$, and a budget  $B \in \mathbb{R}^+$. The goal is to find an out-tree $T$ of $G$  rooted at $r$, such that $c(T)=\sum_{v \in V(T)}c(v) \le B$ and $p(T)=\sum_{x \in X_T} p(x)$ is maximum, where $X_T=\bigcup_{v \in V(T)} S_v$.

Problem \DRCC generalizes several well-known $NP$-hard problems, including the budgeted maximum coverage problem~\cite{khuller1999budgeted}, which is the particular case where the input graph is a bidirected clique; the \emph{Directed Budgeted Node-weighted Steiner} (\DBOM) problem~\cite{d2022approximation,bateni2018improved}, which is the particular case in which each node of the graph is associated with a distinct singleton set and hence $|X|=|V|$; and the \emph{Budgeted Node-weighted Steiner} problem (\BOM), which is the undirected version of \DBOM. Therefore, \DRCC is $NP$-hard to approximate within a factor $1-1/e$ like the budgeted maximum coverage problem~\cite{F98}. Moreover, like \BOM, it admits no $o(\log \log |V|)$-approximation algorithm, unless $NP \subseteq DTIME(n ^{\text{polylog}(n)})$, even if the algorithm is allowed to violate the budget constraint by a factor equal to a universal constant~\cite{kortsarz2009approximating}.

In order to provide a bicriteria approximation algorithm for \DRCC, we will use as a subroutine a polynomial time approximation algorithm for the \emph{node-weighted Directed Steiner tree problem} (\DSteinerT), defined as follows. We are given as input a directed graph $G=(V, A)$, a root node $r\in V$, a set of terminal nodes $R\subseteq V$, and a cost function $c: V \rightarrow \mathbb{R}^{\ge 0}$ defined on the nodes of $G$. The goal is to find an out-tree of $G$ rooted at $r$ and spanning all nodes in $R$, i.e., $R\subseteq V(T)$, such that $c(T)=\sum_{v \in V(T)}c(v)$ is minimum.

Our algorithms will provide a solution with a bounded approximation ratio and a bounded violation of the budget constraint. A polynomial time algorithm is a bicriteria $(\beta, \alpha)$-approximation algorithm if it achieves an approximation ratio of $\alpha>1$ and a budget violation factor of at most $\beta>1$, that is, for any instance $I$ of \DRCC, it returns a solution $T$ such that $p(T) \ge \frac{p(T^*_B)}{\alpha}$ and $c(T) \le \beta B$, where $p(T^*_B)$ is the optimum for $I$ and $B$ is the budget in $I$.

\section{The connected budgeted maximum coverage and the budgeted node-weighted Steiner problems in directed graphs}\label{sec:directed}
In this section, we introduce our approximation algorithms for \DRCC and \DBOM. We start with the polynomial-time bicriteria $\left(1+\epsilon, O\left(\frac{\sqrt{|V|}\log^2{|X|}}{\epsilon^2}\right)\right)$-approximation algorithm for \DRCC, where $\epsilon$ is an arbitrary number in $(0,1]$. We then observe that this algorithm provides a bicriteria $\left(1+\epsilon, O\left(\frac{1}{\epsilon^2}|V|^{1/2}\log^2{|V|}\right)\right)$ for \DBOM, for $\epsilon\in (0,1]$. 

Let $I=<X, \mathcal{S}, G=(V,A), c, p, r, B>$ be an instance of \DRCC. We denote by $T^*_B$ an optimal solution for $I$.

Our algorithm for \DRCC can be summarized in the following three steps: 
\begin{enumerate}
    \item We first define a linear program, denoted as~\eqref{lpDRCC}, whose optimum $OPT$ is an upper bound on the optimum prize $p(T^*_B)$ of $I$.
    \item We give a polynomial-time algorithm that, starting from an optimal solution for~\eqref{lpDRCC}, computes a tree $T$ for which $p(T)=\Omega(OPT)$ and the ratio between prize and cost is $\gamma = \frac{p(T)}{c(T)} = \Omega\left(\frac{OPT}{B\sqrt{|V|}\log^2{|X|}}\right)$.
    \item The cost of $T$ can exceed the budget $B$ but, since the prize-to-cost ratio of $T$ is bounded, we can apply to it a variant of the trimming process given in~\cite{bateni2018improved} to obtain another tree $\hat T$ with cost $\frac{\epsilon}{2}B\le c(\hat T)\le (1+\epsilon)B$, for any $\epsilon \in (0, 1]$, and prize-to-cost ratio $\frac{p(\hat T)}{c(\hat T)}\geq \frac{\epsilon \gamma}{4}$. Therefore, the prize accrued by $\hat{T}$ is  $p(\hat{T})\geq \frac{\epsilon \gamma}{4} c(\hat{T}) = \Omega\left(\frac{\epsilon^2}{\sqrt{|V|}\log^{2}{|X|}} OPT\right)$, which implies an approximation ratio of $O\left(\frac{\sqrt{|V|}\log^{2}{|X|}}{\epsilon^2}\right)$ with a budget violation factor of at most $1+\epsilon$.
\end{enumerate}

Recall that a directed graph $G= (V, A)$ is $B$-proper for a node $r$ if, for every $v\in V$, it holds $dist(r, v)\le B$. Initially, we remove from the input graph all the nodes $v$ having a distance more than $B$ from $r$, making $G$ a $B$-proper graph for $r$.

\subsection*{Upper Bound on the Optimal Prize}

We now provide a linear program whose optimum $OPT$ is an upper bound to the optimum  $p(T^*_B)$ of $I$, i.e., $p(T^*_B) \le OPT$.

We create a directed graph $G'$ in which each node is associated with a cost function $c':V \rightarrow \mathbb{R}^{\ge 0}$ and a prize function $p':V \rightarrow \mathbb{R}^{\ge 0}$. Graph $G'$ is created from $G$ by adding, for each element $x$ of $X$, a node $w_x$ with cost 0 and prize $p(x)$ and, for each node $v\in V$ and each element $x$ that belongs to $S_v$, a directed arc from $v$ to $w_x$. Formally, we let $G' = (V', A')$, where $V' = V \cup W$ with $W=\{w_x:  x\in X\}$ and  $A' = A \cup \{(v, w_x): v\in V, x \in S_v\}$. For each $v\in V$, we let $c'(v)=c(v)$ and $p'(v)=0$, and, for each $w_x\in W$, we let $c'(w_x)=0$ and $p'(w_x)=p(x)$. For each $v\in V'$, we use shortcuts $c_v=c'(v)$ and $p_v=p'(v)$.

For every $v \in V'$, we let $\mathcal{P}_v$ be the set of simple paths in $G'$ from $r$ to $v$. Our linear program~\eqref{lpDRCC} is defined as follows.
\begin{align}\label{lpDRCC}
\text{maximize}\quad\textstyle\sum_{v \in V'} y_v p_v& \tag{LP-DCBC}\\
  \text{subject to}\textstyle\quad \sum_{v \in V'} y_v c_v &\le B\label{lpDRCC:budget}\\
 \textstyle   \sum_{P \in \mathcal{P}_v}f^v_P &= y_v,&& \forall v \in V'\setminus \{r\}\label{lpDRCC:overrallflow}\\
  \textstyle  \sum_{P \in \mathcal{P}_v : z \in P} f^v_{P}&\le y_z,&& \forall z, v \in V'\setminus \{r\}\label{lpDRCC:capacity}\\
 \textstyle   0 \le &y_v \le 1, &&\forall v\in V'\notag\\
 \textstyle   0 \le &f^v_P \le 1, &&\forall v\in V', P \in \mathcal{P}_v .\notag
\end{align}

We use variables $f^v_P$ and $y_v$, for each $v\in V'$ and $P \in \mathcal{P}_v$, where $f^v_P$ is the amount of flow sent from $r$ to $v$ using path $P$ and $y_v$ is the capacity of node $v$ and the overall amount of flow sent from $r$ to $v$.
Variables $y_v$, for $v\in V'$, are called \emph{capacity variables}, while variables $f^v_P$ for $v\in V'$ and $P \in \mathcal{P}_v$ are called \emph{flow variables}. 

The constraints in~\eqref{lpDRCC} are as follows. Constraint~\eqref{lpDRCC:budget} ensures that the (fractional) solution to the LP costs at most $B$. Constraints~\eqref{lpDRCC:overrallflow} and~\eqref{lpDRCC:capacity} formulate the connectivity constraint through a standard flow encoding, that is they ensure that the nodes $v$ with $y_v>0$ induce a subgraph in which all nodes are reachable from $r$. In particular, Constraint~\eqref{lpDRCC:overrallflow} ensures that the amount of flow that is sent from $r$ to a node $v$ is equal to $y_v$ and Constraint~\eqref{lpDRCC:capacity} ensures that the total flow from $r$ to $v$ passing through a node $z$ does not exceed $y_z$.

Note that the number of flow variables is exponential in the size of the input. However, \eqref{lpDRCC} can be solved in polynomial time since, given an assignment of capacity variables, we need to find, independently for any $v \in V'\setminus \{r\}$,  a flow from $r$ to $v$ of overall value $y_v$ that satisfies the capacities of nodes $z \in V'\setminus\{r,v\}$ (see e.g.~\cite{GM93}).

We now show that the optimum $OPT$ of~\eqref{lpDRCC} is an upper bound to the optimum of $I$. In particular, the next lemma shows that, for any feasible solution $T_B$ for $I$, we can compute a feasible solution $\{y_v\}_{v\in V'}$ for~\eqref{lpDRCC} such that $p(T_B)=\sum_{v \in V'}y_v p_v$.

\begin{lemma}\label{lmDirectedLPOtimpality}
Given an instance $I=<X, \mathcal{S}, G=(V,A), c, p, r, B>$ of \DRCC, for any feasible solution $T_B$ for $I$ there exists a feasible solution 
$\{y_v,f^v_P\}_{v\in V', P\in \mathcal{P}_v}$  for~\eqref{lpDRCC} such that $p(T_B)=\sum_{v \in V'}y_v p_v$.
\end{lemma}

\begin{proof}
Let $X_{T_B}=\bigcup_{v \in V(T_B)} S_v$. 
We define a solution to the linear program~\eqref{lpDRCC} in which for all $v\in V(T_B)$ and $x\in X_{T_B}$, we set $y_v=1$ and $y_{w_x} = 1$, while we set $y_u=0$ for any other node $u$ of $V'$. Since $c_w=0$, for all $w\in W$, and $c(T_B)=\sum_{v \in V(T_B)} c_v \le B$, then $\sum_{v \in V'} y_v c_v = \sum_{v \in V} y_v c_v + \sum_{w \in W} y_{w} c_{w} = \sum_{v \in V(T_B)} y_v c_v \leq B$ and the budget Constraint~\eqref{lpDRCC:budget} is satisfied. 

Since $T_B$ is an out-tree, then there exists exactly one path from $r$ to $v$ in $T_B$, for each $v \in V(T_B)$. Let us denote this path by $P_v$.
 For each $x\in X_{T_B}$, let us select an arbitrary $v\in V(T_B)$ such that $x\in S_v$ and let $P_x$ be the path $P_v \cup \{(v,w_x)\}$.
For each $v \in V(T_B)$ and $x\in X_{T_B}$, we set $f^v_{P_v}=1$ and $f^{w_x}_{P_x}=1$, while any other flow variable is set to $0$. Then, Constraints~\eqref{lpDRCC:overrallflow} and~\eqref{lpDRCC:capacity} are satisfied.

Given the definition of $y$ and since $p_v=0$, for all $v\in V$, then $\sum_{v \in V'}y_v p_v =\sum_{v \in V(T_B)}y_v p_v + \sum_{x \in X_{T_B}} y_{w_x} p_{w_x}
=\sum_{x \in X_{T_B}} p_{w_x} = p(T_B)$. This concludes the proof.
\end{proof}

\subsection*{A Tree with a Good Ratio between Prize and Cost}
Here, we give a polynomial time algorithm that computes an out-tree $T$ of $G'$ rooted at $r$, whose prize is $\Omega(OPT)$ and whose ratio between prize and cost is $\Omega\left(\frac{OPT}{B\sqrt{|V|}\log^2{|X|}}\right)$. Note, however, that the cost of $T$ can exceed the budget $B$ by an unbounded factor. We will show in the next section how to trim $T$ in order to bound its cost and, at the same time, retain a good prize.
Here we show the following theorem.

\begin{theorem}\label{th:ratiotree}
There exists a polynomial time algorithm that computes an out-tree $T$ of $G'$ rooted at $r$ such that $p'(T)= \sum_{v\in V(T)}p'(v)$ $= \Omega(OPT)$ and the ratio between prize and cost of $T$ is 
\[
\frac{p'(T)}{c'(T)} =\Omega\left(\frac{OPT}{B\sqrt{|V|}\log^2{|X|}}\right),
\]
where $OPT$ is the optimum of~\eqref{lpDRCC}.
\end{theorem}

To prove the theorem, we start by introducing a polynomial time algorithm that computes an out-tree spanning a given set of nodes, called terminals. The cost of this out-tree is bounded by a function of a lower bound on the amount of flow received by each terminal in an optimal solution for~\eqref{lpDRCC}. Formally, we prove the next lemma. The algorithm in Theorem~\ref{th:ratiotree}, carefully chooses suitable terminal sets that guarantee a lower bound on the obtained prize and on the received flow.

\begin{lemma}\label{lem:cost}
 Let $\{y_v,f^v_P\}_{v\in V', P\in \mathcal{P}_v}$ be an optimal solution for linear program~\eqref{lpDRCC},  $\delta \geq 1$ be a real number, and $R\subseteq W$ be a set of nodes such that $y_w\geq 1/\delta$, for each $w\in R$. Then there exists a polynomial time algorithm that computes an out-tree $T$ of $G'$ rooted at $r$ that spans all the nodes in $R$ and costs $c'(T)=O(\delta B\sqrt{|V|}\log |R|)$.
\end{lemma}
\begin{proof}
    The proof is summarized as follows. We consider the set of nodes in $R$ as the set of terminals in an instance of the node-weighted Directed Steiner tree problem (\DSteinerT) where $r$ is the root node. By using solution $\{y_v,f^v_P\}_{v\in V', P\in \mathcal{P}_v}$ and the lower bound on the amount of flow received by each node in $R$, we show that the optimum for a fractional relaxation of this instance of \DSteinerT is at most $\delta B$. Then, we apply the approximation algorithm for \DSteinerT that we will give in Section~\ref{sec:DSteinerT}, which computes a tree whose cost is a factor $O(\sqrt{|V|}\log |R|)$ from the optimum of the same fractional relaxation. Therefore, we obtain an out-tree that is rooted at $r$, spans all nodes in $R$, and costs $O(\delta B\sqrt{|V|}\log |R|)$, proving the theorem.

    We now give the details of the proof. We first introduce the notation for problem \DSteinerT and its linear relaxation.
    In \DSteinerT, we are given a directed graph $G''=(V'', A'')$ with nonnegative costs assigned to its nodes and a set of terminals $R \subseteq V''$, and the goal is to find an out-tree of $G''$ rooted at the given root node spanning $R$ such that the total cost on its nodes is minimum. We consider the standard flow-based linear programming relaxation of \DSteinerT (called \FDSteinerT) in which we need to assign capacities to nodes in such a way that the total flow sent from the root node to any terminal is $1$ and the sum of node capacities multiplied by their cost is minimized. Formally, given a directed graph $G''=(V'', A'')$, a root node $r\in V''$, a nonnegative node-cost function $c'': V'' \rightarrow \mathbb{R}^{\ge 0}$, and a set of terminals $R \subseteq V''$, \FDSteinerT requires to solve the following linear program.
    \begin{align}\label{lpDSteinerTree}
 \text{minimize}\quad \textstyle\sum_{v \in V''} x_v c_v & \tag{LP-DST}\\
  \text{subject to}\quad\textstyle\sum_{P \in \mathcal{P}_t}g^t_P &= 1, &&\forall t \in R\label{lpDSteinerTree:overallflow}\\
   \textstyle\sum_{P \in \mathcal{P}_t:v \in P} g^t_{P}&\le x_v, &&\forall v \in V'', t \in R\label{lpDSteinerTree:capacity}\\
  \textstyle  0 \le &x_v \le 1, &&\forall v\in V''\notag\\
   \textstyle 0 \le &g^t_P \le 1, &&\forall t\in R, P \in \mathcal{P}_t,\notag
\end{align}
where $\mathcal{P}_t$ is the set of all simple paths from $r$ to $t$ in $G''$, for each $t\in R$, and $c_v = c''(v)$, for each $v\in V''$.
Similarly to~\eqref{lpDRCC}, we use variables $x_v$ and $g^t_P$ as capacity and flow variables, respectively, for each $v\in V''$, $t\in R$, and $P\in \mathcal{P}_t$. As for~\eqref{lpDRCC}, Constraints~\eqref{lpDSteinerTree:overallflow} and~\eqref{lpDSteinerTree:capacity} ensure connectivity, but, differently from~\eqref{lpDRCC}, we require that all terminals receive an amount of flow from $r$ equal to 1, while the other nodes do not need to receive a predefined amount of flow.

    From $G'=(V', A')$ and $R$, we define an instance $I_{DST}$ of \DSteinerT as follows. We create a directed graph $G''=(V'', A'')$ as the subgraph of $G'$ induced by $V''=V \cup R$. The set of terminals in $I_{DST}$ is $R$, the root node is $r$ and the node costs are defined  as $c'$, i.e., $c''(v)=c'(v)$, for each $v\in V''$. 
    Let $I_{FDST}$ be the instance of \FDSteinerT induced by $I_{DST}$ as in~\eqref{lpDSteinerTree} and let $OPT_{FDST}$ be the optimum for $I_{FDST}$.
    
    We now argue that the optimum $OPT_{FDST}$ for $I_{FDST}$ is at most $\delta B$. Starting from the solution $\{y_v,f^v_P\}_{v\in V', P\in \mathcal{P}_v}$ for~\eqref{lpDRCC}, we define a solution $\{x_v,g^t_P\}_{v\in V'', t\in R, P\in \mathcal{P}_t}$ for~\eqref{lpDSteinerTree} as follows: $x_t=1$, for each $t\in R$;  $g^t_P = f^t_P / y_t$, for each $t\in R$ and $P\in \mathcal{P}_t$; and $x_v=\max_{t\in R}\{\sum_{P \in \mathcal{P}_t:v \in P} g^t_{P}\}$, for each $v\in V''\setminus R$. We show that the defined solution is feasible for~\eqref{lpDSteinerTree} and its cost is at most $\delta B$, which implies that $OPT_{FDST}\leq \delta B$. Constraint~\eqref{lpDSteinerTree:overallflow} is satisfied as, by Constraint~\eqref{lpDRCC:overrallflow} of~\eqref{lpDRCC}, we have that, for each $t\in R$, $\sum_{P \in \mathcal{P}_t}f^t_P = y_t$ and hence $\sum_{P \in \mathcal{P}_t}g^t_P = \sum_{P \in \mathcal{P}_t}f^t_P/y_t = 1$. Constraint~\eqref{lpDSteinerTree:capacity} is satisfied, as by definition of $x_v$, it holds $x_v\geq\sum_{P \in \mathcal{P}_t:v \in P} g^t_{P}$, for each $v\in V''$ and $t\in R$. The last two constraints are satisfied by definition of $\{x_v,g^t_P\}_{v\in V'', t\in R, P\in \mathcal{P}_t}$ and by Constraint~\eqref{lpDSteinerTree:overallflow}. The cost of $\{x_v\}_{v\in V''}$ is equal to $\sum_{v \in V''} x_v c_v$. For each $v\in V''\setminus R$, let $t_v$ be the terminal that attains the maximum in the definition of $x_v$, i.e., $t_v:=\arg\max_{t\in R}\{\sum_{P \in \mathcal{P}_t:v \in P} g^t_{P}\}$,  then
    \[
x_v = \sum_{P \in \mathcal{P}_{t_v}:v \in P} g^{t_v}_{P} = \sum_{P \in \mathcal{P}_{t_v}:v \in P} f^{t_v}_{P}/y_{t_v} \leq y_{v}/y_{t_v} \leq \delta y_{v},
    \]
    where the first inequality is due to Constraint~\eqref{lpDRCC:capacity} of~\eqref{lpDRCC} and the last inequality is due to $y_t\geq 1/\delta$ for each node $t\in R$. Moreover, $c_t=0$ for each $t\in R$, because $R\subseteq W$. It follows that $\sum_{v \in V''} x_v c_v = \sum_{v \in V''\setminus R} x_v c_v \leq \delta\sum_{v \in V''\setminus R} y_v c_v  \leq \delta \sum_{v \in V'\setminus R} y_v c_v \leq \delta B$, by Constraint~\eqref{lpDRCC:budget} of~\eqref{lpDRCC}.

    Finally, we apply the algorithm in Section~\ref{sec:DSteinerT}. This algorithm is a polynomial time $O(\sqrt{|V''\setminus R|}\log |R|)$-approximation algorithm for \DSteinerT that, starting from an optimal solution to~\eqref{lpDSteinerTree}, computes a tree $T_{DST}$ rooted at $r$ spanning all the terminals. Moreover, the cost of $T_{DST}$ is at most a factor $O(\sqrt{|V''\setminus R|}\log |R|)$ from the fractional optimum $OPT_{FDST}$, that is 
    \[
    c''(T_{DST})=\sum_{v\in V(T_{DST})}c''(v)=O(\sqrt{|V''\setminus R|}\log |R| )OPT_{FDST},
    \]
    see Theorem~\ref{thDSteinerTree}.\footnote{Here we ignore the term $F=\max_{v\in V}dist(r,v)$  because $G$ is $B$-proper and hence $F\leq B$.}
By applying this algorithm to our instance $I_{DST}$ of \DSteinerT,  we obtain a tree $T_{DST}$ that is rooted at $r$ and spans all the nodes in $R$. The costs of $T_{DST}$ is
    \begin{align*}
    c''(T_{DST}) &=O(\sqrt{|V''\setminus R|}\log |R|)OPT_{FDST} \\
    &= O(\sqrt{|V|}\log |R|)OPT_{FDST} \\
    &= O(\delta B\sqrt{|V|}\log |R|),
    \end{align*}
    as $V''\setminus R=V$ and $OPT_{FDST}\leq \delta B$. This concludes the proof.
    \end{proof}

We now prove Theorem~\ref{th:ratiotree}.
\begin{proof}[Proof of  Theorem~\ref{th:ratiotree}]
We first compute an optimal solution $\{y_v,f^v_P\}_{v\in V', P\in \mathcal{P}_v}$  for the Linear Program~\eqref{lpDRCC}. 
Let $Z \subseteq W$ be the set of nodes in $W$ that in solution $\{y_v,f^v_P\}_{v\in V', P\in \mathcal{P}_v}$ receive at least $\frac{1}{|X|^2}$ amount of flow from $r$, i.e., for any $w \in Z, y_{w}\ge \frac{1}{|X|^2}$. The overall prize accrued by all nodes in $Z$ is
\begin{align*}
\sum_{w \in Z} p_w &\geq \sum_{w \in Z} y_w p_w = OPT-\sum_{w \in W\setminus Z} y_w p_w\\&\ge \left(1-\sum_{w \in W\setminus Z} y_w\right)OPT 
\ge \left(1-|X|\cdot \frac{1}{|X|^2}\right)OPT\\
&=\left(1-\frac{1}{|X|}\right)OPT,\label{EqDirectedPrizeofSelecctedSet}
\end{align*}
where the first inequality holds as $y_w \le 1$, for each $w \in Z$, the second inequality holds as the prize of each node is no more than $OPT$ and the third inequality holds because each node $w \in W\setminus Z$ has $y_w < \frac{1}{|X|^2}$ and $|W\setminus Z| \leq |W| = |X|$.

From now on we only consider the prize accrued by nodes in $Z$, which results in losing a factor of at most $1-\frac{1}{|X|} = \Theta(1)$ with respect to the optimum of~\eqref{lpDRCC}.
To simplify the reading, we ignore this constant factor and assume that $\sum_{w \in Z} y_w p_w = OPT$.

We partition the nodes of $Z$ into $k$ disjoint sets $Z_1, \dots, Z_k$ defined as $Z_i=\left\{w\in Z : y_w\in \Big(\frac{1}{2^i},\frac{1}{2^{i-1}}\Big]\right\}$,  for each $i \in [k]$. It is easy to see that $k=O(\log |X|)$ such sets are enough to cover all nodes of $Z$. In fact, if the smallest value of $y_w$ for a node $w\in Z$ is in the interval $\Big(\frac{1}{2^k},\frac{1}{2^{k-1}}\Big]$, then, since $y_w\geq \frac{1}{|X|^2}$, we have $\frac{1}{2^{k-1}}\ge \frac{1}{|X|^2}$, and hence $2^{k-1} \le |X|^2$ and $k\le 2\log |X|+1$.

We distinguish between two cases by dividing $Z$ into two parts $Z_A=\bigcup_{i=1}^{\lfloor\log\log {|X|}\rfloor} Z_i$ and $Z_B = Z\setminus Z_A=\bigcup_{i=\lfloor\log\log {|X|}\rfloor+1}^{k} Z_i$. Since $\sum_{w \in Z} y_w p_w = OPT$, we must have $\sum_{w \in Z_A}y_w p_w \ge \frac{OPT}{2}$ or $\sum_{w \in Z_B}y_w p_w \geq \frac{OPT}{2}$.
\begin{enumerate}
    \item $\sum_{w \in Z_A}y_w p_w \ge \frac{OPT}{2}$. In this case, we consider the set of nodes in $Z_A$ as the set of terminals $R$ in Lemma~\ref{lem:cost}. Since $y_w\geq 1/2^{\lfloor\log\log {|X|}\rfloor} \geq 1/2^{\log\log{|X|}} = 1/\log{|X|}$, for each $w\in Z_A$, in Lemma~\ref{lem:cost} we can set $\delta=\log{|X|}$. Therefore, by applying the algorithm in Lemma~\ref{lem:cost}, we obtain a tree $T$ rooted at $r$ that spans all the nodes in $Z_A$ and costs $c'(T)= O(B\sqrt{|V|}\log^2 |X|)$. Moreover, as $T$ spans all the nodes in $Z_A$, its prize is at least 
    \[
    \quad\quad\quad p'(T) =\sum_{v\in V(T)}p'(v) \geq \sum_{w \in Z_A}p_w  \geq \sum_{w \in Z_A}y_w p_w \ge \frac{OPT}{2},
    \]
    by the case assumption and monotonicity of the prize function. Therefore, the ratio between prize and cost of $T$ is $\frac{p'(T)}{c'(T)} = \Omega\left(\frac{OPT}{B\sqrt{|V|}\log^2{|X|}}\right)$.

    \item $\sum_{w \in Z_B}y_w p_w \geq \frac{OPT}{2}$. Since $k\leq 2\log{|X|} +1$, there must be an index $i$ between $\lfloor\log\log {|X|}\rfloor+1$ and $2\log{|X|} +1$ such that 
    \[
    \quad\sum_{w \in Z_i} y_w p_w \ge \frac{OPT/2}{2\log{|X|} - \lfloor\log\log {|X|}\rfloor + 1} \ge \frac{OPT}{4\log{|X|}},
    \]
    for $|X|$ sufficiently large.
    Let $p'(Z_i)$ be the sum of prizes of all the nodes in $Z_i$. Then,
    \begin{equation}\label{th:ratiotree:prize2ndcaseeqone}
    \quad\quad p'(Z_i)=\sum_{w \in Z_i} p_w \ge 2^{i-1} \sum_{w \in Z_i} y_w p_w \ge 2^{i-1}\frac{OPT}{4\log{|X|}},    
    \end{equation}
    since $y_w\in \Big(\frac{1}{2^i},\frac{1}{2^{i-1}}\Big]$, for each $w \in Z_i$.
    Moreover, since $i \geq \lfloor\log\log {|X|}\rfloor+1\geq \log\log|X|$, then
    \begin{align}
    p'(Z_i) &\ge 2^{i-1}\frac{OPT}{4\log{|X|}} \geq 2^{\log\log {|X|}-1}\frac{OPT}{4\log{|X|}}\nonumber \\
    & = \frac{\log{|X|}}{2}\frac{OPT}{4\log{|X|}} = \Omega(OPT).\label{th:ratiotree:prize2ndcaseeqtwo}
    \end{align}    
    Similarly to the previous case, we apply the algorithm in Lemma~\ref{lem:cost}, considering $Z_i$ as set of terminals and $\delta = 2^i$, since $y_w \geq 1/2^i$, for each $w\in Z_i$. The tree $T$ computed by the algorithm in the lemma has cost $c'(T)=O(2^i B \sqrt{|V|}\log{|X|})$ and, since it spans all the nodes in $Z_i$, has prize $p'(T)\geq p'(Z_i)\geq 2^{i-1}\frac{OPT}{4\log{|X|}}$, by Inequality~\eqref{th:ratiotree:prize2ndcaseeqone}. Therefore, the prize-to-cost ratio of $T$ is $\frac{p'(T)}{c'(T)} = \Omega\left(\frac{OPT}{B\sqrt{|V|}\log^2{|X|}}\right)$. Moreover, by Inequality~\eqref{th:ratiotree:prize2ndcaseeqtwo},
    $p'(T) \geq p'(Z_i) = \Omega (OPT)$.\qedhere 
\end{enumerate}
\end{proof}

\subsection*{Trimming Process} In the previous step, we computed an out-tree $T$ of $G'$ rooted at $r$ whose prize is $\Omega(OPT)$. If the cost of $T$ satisfies the budget constraint, this gives a constant approximation factor. However, the cost of $T$ can exceed the budget $B$. In this case, we can exploit the fact that the ratio between prize and cost of $T$ is bounded by $\gamma=\frac{p'(T)}{c'(T)}=\Omega\left(\frac{OPT}{B \sqrt{|V|}\log^2{|X|}}\right)$. In fact, this property allows us to use the trimming process introduced in the following lemma by Bateni et al.~\cite{bateni2018improved} for the node-weighted budgeted problem in undirected graphs.
\begin{lemma}[Lemma 3 in \cite{bateni2018improved}]\label{lmBateniTrimmingProcess}
Let $T$ be a tree rooted at $r$ with  prize-to-cost ratio $\gamma=\frac{p(T)}{c(T)}$. Suppose the underlying graph is $B$-proper for $r$ and for $\epsilon \in (0, 1]$ the cost of the tree is at least $\frac{\epsilon B}{2}$. One can find a tree $\hat T$ containing $r$ with prize-to-cost ratio at least $\frac{\epsilon \gamma}{4}$ such that $\epsilon B/2 \le c(\hat T) \le (1+\epsilon)B$.
\end{lemma}
Note that the above lemma has been introduced for (undirected) rooted trees, but it is easy to see that it can be extended to rooted out-trees, see e.g.~\cite{d2022approximation}.
If $c'(T)>B$, we apply to $T$ the trimming process of Lemma~\ref{lmBateniTrimmingProcess} and obtain another out-tree $\hat T$ of $G'$ with cost between $\frac{\epsilon B}{2}$ and $(1+\epsilon)B$ and prize-to-cost ratio $\frac{p'(\hat T)}{c'(\hat T)} \geq \frac{\epsilon \gamma}{4}$, for any $\epsilon\in (0,1]$. 
Tree $\hat T$ violates the budget at most by a factor $1+\epsilon$. Moreover, the prize of $\hat T$ is $p'(\hat T)\geq \frac{\epsilon \gamma}{4}c'(\hat T)=\Omega\left(\frac{\epsilon OPT}{B \sqrt{|V|}\log^2{|X|}} c'(\hat T)\right)$. Since $c'(\hat T) \ge \epsilon B/2$ and $OPT\geq p(T^*_B)$,  then $p'(\hat T)=\Omega\left( \frac{\epsilon^2  p(T^*_B)}{\sqrt{|V|}\log^2{|X|}}\right)$. 

It remains to turn the tree $\hat T$ of $G'$ into a tree of $G$ with the same prize and cost by taking the maximal subtree of $\hat T$ containing only nodes in $V$.
This results in the following theorem.

\begin{theorem}\label{thMainDBudget}
For any $\epsilon \in (0, 1]$, problem \DRCC admits a polynomial time bicriteria $\left(1+\epsilon, O\left(\frac{\sqrt{|V|}\log^2{|X|}}{\epsilon^2}\right)\right)$-approximation algorithm.
\end{theorem}

The following corollary follows since we can reduce any instance of the directed Budgeted Node-weighted Steiner problem (\DBOM) to an instance of \DRCC where each node of the graph is associated with a distinct singleton set and hence $|X|=|V|$.
\begin{corollary}\label{cor:MainDBudget}
For any $\epsilon \in (0, 1]$, problem \DBOM admits a polynomial time bicriteria $\left(1+\epsilon, O\left(\frac{\sqrt{|V|}\log^2{|V|}}{\epsilon^2}\right)\right)$-approximation algorithm.
\end{corollary}


\section{The node-weighted Steiner tree problem in directed graphs}\label{sec:DSteinerT}
In this section, we present a polynomial time approximation algorithm for \DSteinerT with approximation ratio $O(\sqrt{|V|}\log{|V|})$, where $V$ is the set of nodes in the graph. More precisely, the cost of the tree computed by our algorithm is a factor $O(\sqrt{|V\setminus R|}\log{|R|})$ far from the optimum of its standard flow-based linear programming relaxation given in~\eqref{lpDSteinerTree} plus the maximum distance from the root to a node, where $R$ is the set of terminals.
The algorithm is used as a subroutine in the previous section but might be of its own interest. Formally, we show the following theorem.
\begin{theorem}\label{thDSteinerTree}
Problem \DSteinerT admits a $O\left((1+\epsilon)\sqrt{|V\setminus R|}\log{|R|}\right)$-approximation algorithm whose running time is polynomial in the input size and in $1/\epsilon$, for any $\epsilon >0$. Moreover, the cost of the tree computed by the algorithm is $O\left((OPT+F)\sqrt{|V\setminus R|}\log{|R|}\right)$, where $OPT$ is the optimum of~\eqref{lpDSteinerTree} and $F=\max_{v\in V}dist(r,v)$.
\end{theorem}
We prove Theorem~\ref{thDSteinerTree} in what follows. Let $T^*$ be an optimal solution to \DSteinerT.
%
We use the standard flow-based linear programming relaxation for \DSteinerT given in~\eqref{lpDSteinerTree} in Section~\ref{sec:directed}. For the sake of completeness, we report the linear program below.\footnote{Note that here the graph is denoted as $G=(V,A)$ instead of $G''=(V'',A'')$.}
   \begin{align}\label{lpDSteinerTreeprime}
 \text{minimize}\quad \textstyle\sum_{v \in V} x_v c_v & \tag{LP-DST}\\
  \text{subject to}\quad \textstyle\sum_{P \in \mathcal{P}_t}g^t_P &= 1, &&\forall t \in R\label{lpDSteinerTree:overallflowprime}\\
    \textstyle\sum_{P \in \mathcal{P}_t:v \in P} g^t_{P}&\le x_v, &&\forall v \in V, t \in R\label{lpDSteinerTree:capacityprime}\\
  \textstyle  0 \le &x_v \le 1, &&\forall v\in V\notag\\
  \textstyle  0 \le &g^t_P \le 1, &&\forall t\in R, P \in \mathcal{P}_t .\notag
\end{align}

It is easy to see that $OPT$, the optimum for~\eqref{lpDSteinerTreeprime}, provides a lower bound to $c(T^*)$. In fact, the solution to~\eqref{lpDSteinerTreeprime} in which $x_v$ is set to 1 if $v\in V(T^*)$ and 0 otherwise, and $g^t_P$ is set to 1 if $P$ is the unique path from $r$ to $t$ in $T^*$ and to 0 otherwise, is feasible for~\eqref{lpDSteinerTreeprime} and has value $\sum_{v \in V} x_v c_v=c(T^*)$.

Let $\{x_v\}_{v\in V}$ be an optimal solution for~\eqref{lpDSteinerTreeprime} and let $S\subseteq V$ be the set of all nodes $v$ with $x_v >0$.
Let $U \subseteq S$ be the set of all nodes with $x_v \ge \frac{1}{\sqrt{|V\setminus R|}}$ for any $v \in U$. Note that nodes in $R$ and $r$ belong to $U$ since we need to send one unit of flow from $r$ to any terminal by Constraint~\eqref{lpDSteinerTree:overallflowprime}.
We call a terminal $t \in R$ a \emph{cheap terminal} if there exists a path from $r$ to $t$ in $G[U]$. We call a terminal $t \in R$ an \emph{expensive terminal} otherwise. Let $CH$ and $EX$ be the set of all cheap and expensive terminals in $R$, respectively.

We now show that we can compute in polynomial time two trees spanning $CH$ and $EX$, resp., and then we show how to merge the two trees into a single tree with cost $O\left((OPT+F)\sqrt{|V\setminus R|}\log{|R|}\right)$.
We first show how to compute a tree $T^{CH}$ rooted at $r$ spanning all the cheap terminals $CH$ with cost $c(T^{CH})\leq\sqrt{|V\setminus R|}\cdot OPT$.
\begin{lemma}\label{clSpanning-Cheap-Terminals}
There exists a polynomial time algorithm that finds a tree $T^{CH}$ rooted at $r$ spanning all the cheap terminals $CH$ with cost $c(T^{CH})\leq\sqrt{|V\setminus R|}\cdot OPT$.
\end{lemma}
\begin{proof}
By definition, each terminal $t$ in $CH$ is reachable from $r$ through some path $P$ that contains only nodes in $U$, i.e., $V(P) \subseteq U$. 
Thus, we compute a shortest path tree $T^{CH}$ rooted at $r$ in $G[U]$ spanning all cheap terminals. For tree $T^{CH}$ we have $\sum_{v \in V(T^{CH})} x_v c_v \leq \sum_{v \in U} x_v c_v \le \sum_{v \in V} x_v c_v =OPT$.
By definition of $U$, we have $x_v \ge \frac{1}{\sqrt{|V\setminus R|}}$ for any $v \in U$, then 
\[
c(T^{CH})=\!\sum_{v \in V(T^{CH})}\!c_v\leq\sqrt{|V\setminus R|}\sum_{v \in V(T^{CH})} x_v c_v  \leq \sqrt{|V\setminus R|}\cdot OPT.\qedhere
\]
\end{proof}

We next show how to compute in polynomial time a tree $T^{EX}$ rooted at $r$ spanning all the expensive terminals $EX$ with cost $c(T^{EX})=O\left((OPT+F)\sqrt{|V\setminus R|}\log{|R|}\right)$. The algorithm to build $T^{EX}$ can be summarized as follows. We first compute, for each $t\in EX$, the set $X_t$ of nodes $w$ in $S\setminus U$ for which there exists a path $P$ from $w$ to $t$ that uses only nodes in $U\cup \{w\}$, i.e., $V(P)\setminus \{w\} \subseteq U$. Then, we compute a small-size hitting set $X'$ of all $X_t$. Finally, we connect $r$ to the nodes of $X'$ and the nodes of $X'$ to those in $EX$ in such a way that each node $t$ in $EX$ is reached from one of the nodes in $X'$ that hits $X_t$. The bound on the cost of $T^{EX}$ follows from the size of $X'$ and from the cost of nodes in $U$.

\begin{lemma}\label{clSpanning-Expensive-Terminals}
There exists a polynomial time algorithm that finds a tree $T^{EX}$ rooted at $r$ spanning all the expensive terminals $EX$ with cost $c(T^{EX})\leq (OPT+F)\sqrt{|V\setminus R|}\log{|R|}$.
\end{lemma}
\begin{proof}
Let $U' \subseteq S$ be the set of all nodes $v$ with $0<x_v <\frac{1}{\sqrt{|V\setminus R|}}$, i.e., $U'=S\setminus U$. 
Recall that for any expensive terminal $t \in EX$, we define $X_t$ as the set of nodes $w$ in $U'$ such that there exists a path from $w$ to $t$ in $G[U \cup \{w\}]$.

We first show a lower bound on the size of sets $X_t$, for each $t \in EX$, which will allow us to compute a small hitting set of all such sets. 
\setcounter{claim}{3}
\begin{claim}\label{clBound-X-t}
$|X_t| \ge \sqrt{|V\setminus R|}$, for each $t \in EX$.
\end{claim}
\begin{proof}
We know that (i) each terminal must receive one unit of flow (by Constraint~\eqref{lpDSteinerTree:overallflowprime} of~\eqref{lpDSteinerTreeprime}), (ii) any path $P$ from $r$ to any $t \in EX$ in the graph $G[S]$ contains at least one node $w \in U'$ (by definition of expensive terminals), and, (iii) in any path $P$ from $r$ to any $t \in EX$, the node $w \in U'$ in $P$ that is closest to $t$ is a member of $X_t$, i.e., $w \in X_t$ (by definition of $X_t$), therefore any flow from $r$ to $t$ must pass through a node $w \in X_v$. This implies that the nodes in $X_t$ must send one unit of flow to $t$ in total. Since each of them can only send at most $1/\sqrt{|V\setminus R|}$ amount of flow, they must be at least $\sqrt{|V\setminus R|}$. Formally, we have 
\begin{align*}
 1 &= \sum_{P\in\mathcal{P}_t} g^t_P \leq \sum_{w \in X_t}\sum_{P\in\mathcal{P}_t:w\in P} g^t_P\\
 &\leq \sum_{w \in X_t} x_w < \sum_{w \in X_t} \frac{1}{\sqrt{|V\setminus R|}} =\frac{|X_t|}{\sqrt{|V\setminus R|}},
\end{align*}
which implies that $|X_t| \ge \sqrt{|V\setminus R|}$. The first equality follows from Constraint~\eqref{lpDSteinerTree:overallflowprime} of~\eqref{lpDSteinerTreeprime}, the first inequality is due to the fact that, by definition of $X_t$, any path $P$ from $r$ to a $t \in EX$ contains a node $w \in X_t$, the second inequality is due to Constraint~\eqref{lpDSteinerTree:capacityprime} of~\eqref{lpDSteinerTreeprime}, the last inequality
is due to $X_t\subseteq U'$ and $x_w<\frac{1}{\sqrt{|V\setminus R|}}$, for each $w\in U'$.
This concludes the proof of the claim.
\end{proof}

We use the following well-known result (see, e.g., Lemma~3.3 in~\cite{Chan2007})
to find a small set of nodes that hits all sets $X_t$, for all $t\in EX$. %
\begin{claim}\label{clHittingSet}
Let $V'$ be a set of $M$ elements and $\textstyle\sum=(X'_1, \dots, X'_N)$ be a collection of subsets of $V'$ such that $|X'_i|\ge L$, for each $i\in [N]$.
There is a deterministic algorithm that runs in polynomial time in $N$ and $M$ and finds a subset $X' \subseteq V'$ with $|X'|\le (M/L)\ln{N}$ and $X' \cap X'_i\ne \emptyset$ for all $i\in[N]$.
\end{claim}

Thanks to Claim~\ref{clBound-X-t}, we can use the algorithm of Claim~\ref{clHittingSet} to find a set $X'\subseteq \bigcup_{t \in EX}X_t$ such that $X' \cap X_t\ne \emptyset$, for all $t\in EX$, whose size is at most $|X'| \le \frac{|V\setminus R|\log{|R|}}{\sqrt{|V\setminus R|}} =\sqrt{|V\setminus R|}\log{|R|}$, where the parameters of Claim~\ref{clHittingSet} are $L=\sqrt{|V\setminus R|}$, $N=|EX|\leq |R|$, and $M=\big|\bigcup_{t \in EX}X_t\big| \leq |V\setminus R|$, since $x_t=1$, for each $t\in R$, and hence no node in $R$ can belong to $\bigcup_{t \in EX}X_t \subseteq U'$.

Since for any $t \in EX$ and any $w \in X_t$ there exists a path from $w$ to $t$ in $G[U \cup \{w\}]$ and $X' \cap X_t\ne \emptyset$, then there exists at least a node $w\in X'$ for which there is a path from $w$ to $t$ in $G[U \cup \{w\}]$.

Now, for each $w \in X'$, we find a shortest path from $r$ to $w$ in $G$. Let $\mathcal{P}_1$ be the set of all these shortest paths. 
We also select, for each $t\in EX$, an arbitrary node $w$ in $X'\cap X_t$ and compute a shortest path from $w$ to $t$ in $G[U\cup \{w\}]$. Let $\mathcal{P}_2$ be the set of all these shortest paths. 
Let $V(\mathcal{P}_1)$ and $V(\mathcal{P}_2)$ denote the union of all nodes of the paths in $\mathcal{P}_1$ and $\mathcal{P}_2$, respectively, and let $G^{EX}$ be the graph induced by all the nodes in $V(\mathcal{P}_1)\cup V(\mathcal{P}_2)$.
We compute a tree $T^{EX}$ rooted at $r$ spanning $G^{EX}$. Note that such a tree exists as in $G^{EX}$ we have for each $w \in X'$ a path from $r$ to $w$ and, for each terminal $t\in EX$, at least a path from one of the nodes in $X'$ to $t$.

We next move to bound the cost of $T^{EX}$; Indeed, we bound the cost of all nodes in $G^{EX}$. 
Since $|X'|\le\sqrt{|V\setminus R|}\log{|R|}$ and $dist(r,v) \leq F$ for any node $v$,
then $c(V(\mathcal{P}_1))\leq F\sqrt{|V\setminus R|}\log{|R|}$. 
Since $x_v \ge \frac{1}{\sqrt{|V\setminus R|}}$ for any $v \in U$, and $\sum_{v \in U}x_v c_v\leq \sum_{v \in S}x_v c_v \leq OPT$, then $c(U)=\sum_{v \in U}c_v\leq \sqrt{|V\setminus R|} \cdot OPT$. Therefore, since $V(\mathcal{P}_2)\setminus X'\subseteq U$, then $c(V(\mathcal{P}_2)\setminus X')\leq c(U)\leq \sqrt{|V\setminus R|} \cdot OPT$.
Overall, $G^{EX}$ costs at most $(OPT+F)\sqrt{|V\setminus R|}\log{|R|}$. This finishes the proof of Lemma~\ref{clSpanning-Expensive-Terminals}.
\end{proof}

We can now prove Theorem~\ref{thDSteinerTree}.

\begin{proof}[Proof of Theorem~\ref{thDSteinerTree}]
Since both $T^{EX}$ and $T^{CH}$ are rooted at $r$, we can find a tree $T$ rooted at $r$ spanning all nodes $V({T^{EX}})\cup V({T^{CH}})$.

By Lemmas~\ref{clSpanning-Cheap-Terminals} and~\ref{clSpanning-Expensive-Terminals} we have that the cost of $T$ is 
$c(T)=O\left((OPT + F)\sqrt{|V\setminus R|}\log{|R|}\right)$. This shows the second part of the statement.

To show the bound on the approximation ratio, we observe that $OPT\leq c(T^*)$. Moreover, we can assume that $F\leq (1+\epsilon)c(T^*)$ since we can remove from the graph all the nodes $v$ such that $dist(r, v) > (1+\epsilon)c(T^*)$
by estimating the value of $c(T^*)$ using a binary search (see~\ref{apx:binarysearch} for more details).
Therefore, the cost of $T$ is $c(T)=O\left(\sqrt{|V\setminus R|}\log{|R|}\right)c(T^*)$.
\end{proof}

\section{Discussion and future research}
\DRCC and \CBMC are basic combinatorial optimization problems with many applications in diverse areas such as logistics, wireless sensor networks, and bioinformatics. Besides their relevance, their approximation properties still need to be better understood. In this paper, we provide an approximation algorithm for \DRCC with a sublinear approximation ratio. Our results also imply an improved approximation for the particular case of additive prize function (\DBOM).

The most interesting but very ambitious research question is whether there is a polynomial lower bound on the approximability of \DRCC. In other words, whether it is hard to compute in polynomial time a solution that is asymptotically better than a polynomial factor from the optimum. 
The same question for the directed Steiner tree problem has been open for a long time. However, it is known that the integrality gap of the standard flow-based LP relaxation for \DRCC is unbounded if no budget violation is allowed~\cite{bateni2018improved} and has a polynomial lower bound for the directed Steiner tree problem~\cite{li2025polynomial}. This suggests that we cannot significantly improve our approximation factors for \DRCC by using the linear relaxation~\eqref{lpDRCC}. 
Using LP-hierarchies~\cite{rothvoss2011directed, friggstad2014linear} could be a promising research direction to improve our approximation factors. For the Directed Steiner Network, it is known that the integrality gap of the Lasserre Hierarchy has a polynomial lower bound~\cite{DinitzNZ20}.
An even harder research question is to find a lower bound on the approximation of \CBMC.

The techniques introduced in this paper might be useful to approximate other more general network design problems. 
One interesting example is the case where the prize function is a monotone submodular set function of the nodes. In this case, the best algorithm is the one in~\cite{d2022budgeted} that achieves an approximation factor of $O(\frac{1}{\epsilon^{3}}\sqrt{B})$-approximation algorithm with a budget violation of a factor $1+\epsilon$, for any $\epsilon\in (0,1]$.
Our algorithms cannot directly be applied to this case because the linear program~\eqref{lpDRCC} does not give an upper bound to the optimum. Therefore, the first step in using our techniques should be to find a suitable linear relaxation.

\bibliographystyle{abbrv} 
\bibliography{references}

\appendix

\section*{Appendix}

\section{Details on the estimation for the optimum directed Steiner tree}\label{apx:binarysearch}

Let $c_{\min}$ be the minimum positive cost of a vertex and $c_{M}$  be the cost of all nodes in $G$, that is, $c_{M} = \sum_{v\in V}c(v)$. We know that $c(T^*)\leq c_{M}$. We estimate the value of $c(T^*)$ by guessing $N$ possible values, where $N$ is the smallest integer for which $c_{\min}(1+\epsilon)^{N-1}\geq c_{M}$.

For each guess $i\in [N]$, we remove the nodes $v$ with $dist(r, v)> c_{\min}(1+\epsilon)^{i-1}$
, and compute a Steiner Tree in the resulting graph, if it exists, with the algorithm in Section~\ref{sec:DSteinerT}. Eventually, we output the computed Steiner Tree with the smallest cost.

Since $c_{\min}(1+\epsilon)^{N-2}< c_{M}$, the number $N$ of guesses is smaller than $\log_{1+\epsilon}(c_{M}/c_{\min}) + 2$, which is polynomial in the input size and in $1/\epsilon$.

Let $i\in [N]$ be the smallest value for which $c_{\min}(1+\epsilon)^{i-1}\geq c(T^*)$. Then, $c(T^*)> c_{\min}(1+\epsilon)^{i-2}$ and for each node $v$ in the graph used in guess $i$, we have $dist(r, v)\leq c_{\min}(1+\epsilon)^{i-1} <(1+\epsilon)c(T^*)$. 
Since we output the solution with the minimum cost among those computed in the guesses for which our algorithm returns a feasible Steiner Tree, then the final solution will not be worse than the one computed at guess $i$. Hence, in Section~\ref{sec:DSteinerT} we focus on guess $i$ and assume that $dist(r, v) \leq (1+\epsilon)c(T^*)$
, for all  nodes $v\in V$.

\end{document}